\documentclass{article}

\usepackage[T1]{fontenc}
\usepackage{xcolor}
\usepackage{amsmath}
\usepackage{amsfonts}
\usepackage{fullpage}
\usepackage{tikz}
\usepackage{amsthm}
\usepackage[short]{optidef}
\usepackage{mathtools}
\usepackage{csquotes}
\usepackage{hyperref}

\newtheorem{theorem}{Theorem}
\newtheorem{lemma}{Lemma}
\newtheorem{question}{Question}
\theoremstyle{definition}
\newtheorem{definition}{Definition}

\title{Matching Markets with Chores}

\author{Jugal Garg\thanks{Supported by NSF grants CCF-1942321 and CCF-2334461.} \and Thorben Tr\"obst\thanks{Supported by NSF grant CCF-2230414.} \and Vijay V.\ Vazirani\footnotemark[2]}

\begin{document}

\maketitle

\begin{abstract}
The fair division of chores, as well as mixed manna (goods and chores), has received substantial recent attention in the fair division literature; however, ours is the first paper to extend this research to matching markets. Indeed, our contention is that matching markets are a natural setting for this purpose, since the manna that fit into the limited number of
hours available in a day can be viewed as one unit of allocation. We extend several well-known results that hold for goods  to the settings of chores and mixed manna. 
In addition, we show that the natural notion of an earnings-based equilibrium, which is more
natural in the case of all chores, is equivalent to the pricing-based equilibrium
given by Hylland and Zeckhauser for the case of goods. 
\end{abstract}

\section{Introduction}
In a \emph{one-sided matching market}, we are given a set $A$ of \emph{agents} and a set $G$ of
\emph{items} (traditionally goods).
Each agent has preferences over the items.
We assume that $|A| = |G| = n$ and the goal is then to find a perfect matching between items and
agents which satisfies certain desirable properties including fairness and efficiency.

It is important to note that we do not allow monetary transfers.
Markets of this kind arise in various situations in which we want to fairly allocate items/entities among
people but in which payments would be considered immoral or impractical.
For example, consider assigning students to schools (or to individual courses), assigning organ
donations to recipients, or doctors to hospitals.

One can distinguish these markets based on whether the preferences of the agents are ordinal, i.e.\
each agent submits a total order over the items, or cardinal, i.e.\ each agent $i$ submits numerical utilities
$(u_{i j})_{j \in G}$. Whereas ordinal preferences are easier to elicit, cardinal preferences are more expressive, thereby producing higher quality allocations and leading to significant gain in efficiency, e.g., \cite{ILWM17}  give a striking example with $n$ agents and goods, in which an allocation under cardinal utilities improves every agent by a factor of $\Theta(\log n)$ over the (coarse) allocation made under ordinal information.  
For this reason, we are interested in mechanisms for matching markets which take advantage of
cardinal utilities.

Integral allocations in matching markets cannot achieve any
reasonable measure of fairness, e.g., consider a scenario in which all agents like only one item. For this reason it is customary to allow lotteries over matchings, i.e.\ fractional perfect
matchings, instead. In this setting, the classic mechanism is due to Hylland and Zeckhauser (HZ) \cite{HZ79} based on competitive equilibrium; we will cover this mechanism and related definitions in Section~\ref{sec:eq_notions}.
It finds allocations (lotteries) which are Pareto-optimal (PO) and envy-free (EF). Moreover, it is incentive compatible in the large \cite{HMPY18}.
Unfortunately, the problem of approximating the HZ equilibrium is PPAD-complete \cite{VY21,CCPY22}, 
and in fact so is the more
general problem of finding any EF+PO allocation \cite{CGMM22,TV24}.

Several extensions of the HZ mechanism can be found in the literature, for example for markets with
endowments \cite{EMZ21,GTV24} or in two-sided markets \cite{M16}.
However, the computational difficulties remain.
To counter this, \cite{HV22} proposed Nash-bargaining-based matching market models. 
This leads to polynomial time mechanisms \cite{HV22,PTV21} which are PO and
\emph{approximately} EF \cite{TV24}.

In all of these cases, all utilities are assumed to be \emph{non-negative}, i.e.\ 
$G$ consists exclusively of \emph{goods}. While considering chores as well, 
three types of settings need to be studied:  
\begin{enumerate}
    \item In the \emph{goods} setting, we have $u_{i j} \geq 0$ for all $i, j$.
    \item In the \emph{chores} setting, we have $u_{i j} \leq 0$ for all $i, j$.
    \item The \emph{mixed} setting, also called \emph{mixed manna}. 
\end{enumerate}

For a motivating consider the division
of various tasks or activities among people.
An activity might be enjoyable ($u_{ij} > 0$) or it might be a chore ($u_{ij} < 0$).
Indeed, it may even be enjoyable for some $i$ and a chore for another $i'$.
The matching constraints on the agents enforce the fact that each agent has only a limited number of
hours available in a day to do both enjoyable and displeasing activities.

The fair division of chores as well as mixed manna has received substantial recent attention in the fair division literature; see \cite{GLD23,amanatidis2023fair,liu2024mixed} for recent surveys. Ours is first paper to extend this research to the area of matching markets. Indeed, our contention is that matching markets are a natural setting for this purpose, since the manna that fit into the limited number of
hours available in a day can be viewed as one unit of allocation. 
Our contributions are the following:
\begin{enumerate}
    \item We make the observation that, due to the matching constraints, many results can be easily
        translated between the various settings; see Section~\ref{sec:shifting}.
    \item We consider an earnings-based equilibrium analogous to the equilibrium proposed for the
        fair division of chores by Bogomolnaia et al.\ \cite{BMSY17}.
        We show that for matching markets, this equilibrium notion coincides with the HZ
        equilibrium; see Section~\ref{sec:eq_equiv}.
    \item We show that an EF+PO allocation can be found in polynomial time if we restrict ourselves
        to two types of agents, where each agent's utility function is one of two given functions; see Section~\ref{sec:const_types}.
    \item We study potential extensions of the Nash-bargaining-based mechanism for chores.
        We give several counterexamples that natural extensions to chores are highly inefficient;
        see Section~\ref{sec:nash_bargaining}.
\end{enumerate}

Our paper exposes an interesting direction for future research, namely the problem of finding
approximately fair and efficient allocations in the chores setting in polynomial time.

\subsection{Additional Related Work}
Matching markets have been extensively studied with numerous  applications across various multi-agent settings (see e.g.,~\cite{AzizGSW19,AzizCS23,ChoKLLL0YY24,BeynierMRS21,GuptaPSZ19,HosseiniLC18,Aziz20b}).

In the ordinal setting, two popular mechanisms are Probabilistic Serial~\cite{Bogomolnaia-PS} and Random Priority~\cite{Moulin2018fair}. Random Priority is strategyproof but lacks envy-freeness, while Probabilistic Serial gurantees envy-freeness but is not strategyproof. 

Several researchers have proposed HZ-type mechanisms for a variety of applications, including~\cite{Budish2011combinatorial, He2018pseudo, Le2017competitive, Mclennan2018efficient}. The basic framework has also been extended in several directions, such as two-sided matching markets, incorporating quantitative constraints, and adapting to settings where agents have initial bundles of items rather than money; see e.g.,~\cite{EMZ21, Echenique2019fairness}.

Competitive equilibrium with equal incomes (CEEI) is a classic approach for achieving envy-free and Pareto optimal allocations~\cite{VARIAN74}. This method has been extensively studied in the context of both goods and chores, as well as for mixed manna, particularly in settings without matching constraints; see recent papers~\cite{GargM20,CKMN24} and the references therein.

\section{Equilibrium Notions}\label{sec:eq_notions}
A standard approach to finding desirable allocations in the all goods setting is the
celebrated Hylland Zeckhauser (HZ) mechanism \cite{HZ79}.
In this section, we will briefly review the workings of the HZ mechanism.

The core idea is to implement a \emph{pseudo-market}, i.e.\ to introduce some amount of fake money
to create a market with money and then to use a market equilibrium in order to find our desirable
allocations.
In order for this to be fair, in the HZ mechanism, each agent gets exactly one unit of fake money.
The corresponding equilibrium notion is given below.

\begin{definition}
    An \emph{HZ equilibrium} consists of an allocation $x = (x_{ij})_{i \in A, j \in G}$ and prices
    $p = (p_j)_{j\in G}$ (both non-negative) such that:
    \begin{enumerate}
        \item Each agent is fully matched, i.e.\ for each $i\in A$, the sum of allocations across all goods $j\in G$ satisfies $\sum_{j \in G} x_{i j} = 1$.
        \item Each good is fully matched, i.e.\ for each $j\in G$, the sum of allocations across all agents $i\in A$ satisfies $\sum_{i \in A} x_{i j} = 1$.
        \item No agent overspends, i.e.\ $p \cdot x_i = \sum_{j \in G} p_j x_{i j} \leq 1$.
        \item Every agent gets a cheapest optimal bundle, i.e.\ for any $(y_j)_{j \in G}$ with $y
            \geq 0$, $\sum_{j \in G} y_j = 1$, and $p \cdot y \leq 1$, we have that $u_i \cdot y
            \leq u_i \cdot x_i$ and if $u_i \cdot y = u_i \cdot x_i$, then $p \cdot y \geq p \cdot x_i$. We note that the condition that the bundles be \emph{cheapest}
            is necessary for Pareto-optimality.
    \end{enumerate}
\end{definition}

\begin{theorem}[Hylland and Zeckhauser \cite{HZ79}]
    An HZ equilibrium always exists.
\end{theorem}

The proof of this fact is non-constructive and relies on Kakutani's fixed point theorem.
Two properties which we are particularly interested in are \emph{Pareto-optimality} and
\emph{envy-freeness} which are the standard notions of efficiency and fairness respectively.

\begin{definition}
    A fractional perfect matching (FPM) $x$ is \emph{Pareto-better} than another FPM $y$
    (alternatively, \emph{Pareto-dominates} $y$) if $u_i \cdot x_i \geq u_i \cdot y_i$ for all $i
    \in A$ and $u_i \cdot x_i > u_i \cdot y_i$ for some $i \in a$.
    We say that an FPM $x$ is \emph{Pareto-optimal (PO)} if no other FPM is Pareto-better than it.
\end{definition}

\begin{definition}
    In an FPM $x$, we say that $i \in A$ envies $i' \in A$ if $u_i \cdot x_{i'} > u_i \cdot x_i$.
    We say that $x$ is \emph{envy-free (EF)} if no agent envies any other agent.
\end{definition}

\begin{theorem}[Hylland and Zeckhauser \cite{HZ79}]
    Let $(x, p)$ be an HZ equilibrium, then $x$ is envy-free and Pareto-optimal (EF+PO).
\end{theorem}

Neither $x$ nor $p$ in an HZ equilibrium need to be unique.
However, it is possible to restrict the equilibrium prices via the following useful lemma.

\begin{lemma}\label{lem:hz_price_zero}
    Let $(x, p)$ be an HZ equilibrium, then there are always prices $p'$ such that $(x, p')$ is an
    HZ equilibrium and there is at least one $j \in G$ with $p'_j = 0$.
\end{lemma}

Inspired by recent work on fair division and market equilibria with chores
\cite{BMSY17,CGMM22,CKMN24}, we define an analogous equilibrium notion in which agents need to
\emph{earn} money rather than \emph{spend} it.
This is a more natural equilibrium in the chores setting since we would want to compensate agents
for doing chores.

\begin{definition}
    An \emph{HZ earnings equilibrium} consists of an allocation $(x_{i j})_{i \in A, j \in G}$ and
    payments $(p_j)_{j \in G}$ (both non-negative) such that:
    \begin{enumerate}
        \item Each agent is fully matched, i.e.\ $\sum_{j \in G} x_{i j} = 1$ for all $i \in A$.
        \item Each good is fully matched, i.e.\ $\sum_{i \in A} x_{i j} = 1$ for all $j \in G$.
        \item No agent under-earns, i.e.\ $p \cdot x_i \geq 1$.
        \item Every agent gets a highest-earning, optimal bundle, i.e.\ for any $(y_j)_{j \in G}$ with $y
            \geq 0$, $\sum_{j \in G} y_j = 1$, and $p \cdot y \geq 1$, we have that $u_i \cdot y
            \leq u_i \cdot x_i$ and if $u_i \cdot y = u_i \cdot x_i$, then $p \cdot y \leq p \cdot
            x_i$.
    \end{enumerate}
\end{definition}

Note that in neither definition did we make the assumption that $u \geq 0$ or $u \leq 0$. As we will show in Section~\ref{sec:shifting}, utilities can be shifted between these settings without loss of generality. Therefore, such assumptions are neither necessary nor helpful.
We can show that HZ earnings equilibria satisfy the same desirable properties as traditional HZ equilibria.

\begin{theorem}
    Let $(x, p)$ be an HZ earnings equilibrium.
    Then $x$ is envy-free.
\end{theorem}

\begin{proof}
    Consider agents $i$ and $i'$.
    The bundle $x_{i'}$ is a feasible bundle for agent $i$ that would provide them with 1 unit of currency. Since $i$ is receiving a utility-maximizing bundle, it follows that $u_i \cdot x_i \geq u_i \cdot x_{i'}$, meaning agent $i$ does not envy $i'$. 
    
    Since $i$ and $i'$ were chosen arbitrarily, this implies that the allocation $x$ is envy-free. 
\end{proof}

\begin{theorem}
    Let $(x, p)$ be an HZ earnings equilibrium.
    Then $x$ is Pareto-optimal.
\end{theorem}

\begin{proof}
    Assume otherwise and let $y$ be a Pareto-better FPM.
    Consider now how much the agents would be earning under $y$ with the prices $p$.
    For each $i \in A$ we have $u_i \cdot y_i \geq u_i \cdot x_i$ since $y$ is Pareto-better than $x$.
    Thus, $p \cdot y_i \leq p \cdot x_i$ as $x_i$ is a \emph{highest-earning}, optimal bundle.

    Moreover, there is at least one $i \in A$ with $u_i \cdot y_i > u_i \cdot x_i$.
    But then we must have $p \cdot y_i < p \cdot x_i$.
    Overall, this shows that \[\sum_{i \in A} p \cdot y_i < \sum_{i \in A} p \cdot x_i .\]

    On the other hand, a simple calculation shows
    %\[\sum_{i \in A} p \cdot y_i = \sum_{i \in A} \sum_{j \in G} p_j y_{i j} = \sum_{j \in G} p_j \sum_{i \in A} y_{i j} = \sum_{j \in G} p_j \]
    \begin{equation}\notag
    \begin{aligned}
    \sum_{i \in A} p \cdot y_i & = \sum_{i \in A} \sum_{j \in G} p_j y_{i j}\\
    & = \sum_{j \in G}
        p_j \sum_{i \in A} y_{i j}\\
        &= \sum_{j \in G} p_j
    \end{aligned}
    \end{equation}
    using only the fact that $y$ is a fractional perfect matching.
    The same calculation holds for $x$ and so we actually have that $\sum_{i \in A} p \cdot y_i =
    \sum_{i \in A} p \cdot x_i$, which is a contradiction. 
\end{proof}

\section{Shifting Utilities}\label{sec:shifting}

Our first observation is that due to the fact that we are considering only fractional perfect
matchings, most interesting properties will be preserved by a simple shifting operation.

\begin{definition}
    Let $(u_{i j})_{i \in A, j \in G}$ and $(c_i)_{i \in A}$ be rational numbers.
    Then we define $(u + c)_{i j} \coloneqq u_{i j} + c_i$.
\end{definition}

\begin{lemma}
    Let $x$ be an FPM which is envy-free wrt.\ utilities $u$.
    Then $x$ is also envy-free wrt.\ utilities $u + c$.
\end{lemma}

\begin{proof}
    Consider agents $i$ and $i'$.
    Then
 %   \[(u + c)_i \cdot x_i = u_i \cdot x_i + c_i \geq u_i \cdot x_{i'} + c_i = (u + c)_i \cdot x_{i'}.    \]
    \begin{equation}\notag
    \begin{aligned}
        (u + c)_i \cdot x_i & = u_i \cdot x_i + c_i \\
        & \geq u_i \cdot x_{i'} + c_i \\
        & = (u + c)_i \cdot x_{i'}.
    \end{aligned}
    \end{equation}
    The only thing we used here is the fact that both $x_i$ and $x_{i'}$ sum up to 1. 
\end{proof}

\begin{lemma}
    Let $x$ be an FPM which is Pareto-optimal wrt.\ utilities $u$.
    Then $x$ is also Pareto-optimal wrt.\ utilities $u + c$.
\end{lemma}

\begin{proof}
    Assume otherwise, i.e.\ let there be an FPM $y$ which is Pareto-better than $x$ wrt.\ utilities
    $u + c$.
    Now observe that
    \[
        u_i \cdot y_i - u_i \cdot x_i = (u + c)_i \cdot y_i - (u + c)_i \cdot x_i
    \]
    again using the fact that $y_i$ and $x_i$ both sum up to 1.
    But from this we can immediately conclude that $y$ is also Pareto-better than $x$ wrt.\ u which
    is a contradiction to the Pareto-optimality of $x$. 
\end{proof}

\begin{lemma}\label{lem:hz_shift}
    Let $(x, p)$ be an HZ equilibrium wrt.\ utilities $u$.
    Then $(x, p)$ is also an HZ equilibrium wrt.\ utilities $u + c$.
\end{lemma}

\begin{proof}
    Clearly $x$ satisfies the first three requirements of being an HZ equilibrium under utilities $u
    + c$ since those are independent of the utilities in the first place.

    Consider some agent $i \in A$ and an alternative bundle $(y_j)_{j \in G}$ with $y \geq 0$,
    $\sum_{j \in G} y_j = 1$, and $p \cdot y \leq 1$.
    Then
    %\[ (u + c)_i \cdot y = u_i \cdot y + c_i \leq u_i \cdot x_i + c_i = (u + c)_i \cdot x_i    \]
    \begin{equation}\notag
    \begin{aligned}
    (u + c)_i \cdot y &= u_i \cdot y + c_i\\
    & \leq u_i \cdot x_i + c_i\\
    & = (u + c)_i \cdot x_i
    \end{aligned}
    \end{equation}
    using that both $y$ and $x_i$ sum up to 1.
    Thus $x_i$ is an optimal bundle for $i$.
    A similar argument shows that $x_i$ is also cheapest. 
\end{proof}

This justifies the previous comment that we need not restrict HZ equilibria to the all goods
setting.
For example, in order to show existence of HZ equilibria for arbitrary utilities, simply shift
everything into the non-negative, get an HZ equilibrium, and then use Lemma~\ref{lem:hz_shift} to
show that said equilibrium is also an HZ equilibrium under the original, mixed utilities.
The same observation holds for HZ earnings equilibria.

\begin{lemma}
    Let $(x, p)$ be an HZ earnings equilibrium wrt.\ utilities $u$.
    Then $(x, p)$ is also an HZ earnings equilibrium wrt.\ utilities $u + c$.
\end{lemma}

\begin{proof}
    Same as proof of Lemma~\ref{lem:hz_shift}. 
\end{proof}

Finally, we remark that HZ envy-freeness, Pareto-optimality, and HZ equilibria are also preserved
under scaling by any positive constant.
And so any transformation of the kind $u \mapsto a u + c$ with $a > 0$ preserves these notions.

\section{Equivalence of Equilibria}\label{sec:eq_equiv}

We will now show that HZ equilibria and HZ earnings equilibria are in fact the same thing.

\begin{theorem}\label{thm:equilibria_equiv}
    Let $x$ be some FPM.
    If there exist prices $(p_j)_{j \in G}$ making $(x, p)$ and HZ equilibrium, then there also
    exist earnings $(q_j)_{j \in G}$ making $(x, q)$ an HZ earnings equilibrium and vice versa.
\end{theorem}

\begin{proof}
    Let $(x, p)$ be an HZ equilibrium.
    For each $j \in G$, we define
    \[
        q_j \coloneqq \frac{p_\mathrm{max} - p_j}{p_\mathrm{max} - 1}
    \]
    where $p_\mathrm{max}$ is the maximum price among all $p_j$.

    First, observe that $q$ is indeed non-negative since $p_j \leq p_\mathrm{max}$.
    Moreover, no agent under-earns since
 %   \[ q \cdot x_i = \frac{p_\mathrm{max} - p \cdot x_i}{p_\mathrm{max} - 1} \geq \frac{p_\mathrm{max} - 1}{p_\mathrm{max} - 1} = 1    \]
    \begin{equation}\notag
    \begin{aligned}
        q \cdot x_i &= \frac{p_\mathrm{max} - p \cdot x_i}{p_\mathrm{max} - 1}\\
        & \geq  \frac{p_\mathrm{max} - 1}{p_\mathrm{max} - 1}\\
        & = 1
    \end{aligned}
    \end{equation}
    where we used that $x_i$ sums up to 1 and that $p \cdot x_i \leq 1$.

    In fact, what this calculation shows is that any bundle which costs at most 1 wrt.\ p, earns at
    least 1 wrt.\ q.
    So in fact, this shows that $x$ is an optimum bundle wrt.\ earnings $q$ as it is an optimum
    bundle wrt.\ prices $p$.
    Lastly, observe that if $(y_j)_{j \in G}$ and $(z_j)_{j \in G}$ are non-negative and each sum to
    1, then $p \cdot y \leq p \cdot z$ if and only if $q \cdot y \geq q \cdot z$.
    So $x$ is a highest-earning optimum bundle under $q$ since it was a cheapest optimum bundle
    under $p$.
    
    We have just shown that $(x, q)$ is an HZ earnings equilibrium.
    The proof in the other direction is nearly identical.
    Given an HZ earnings equilibrium $(x, q)$, we define
    \[
        p_j \coloneqq \frac{q_\mathrm{max} - q_j}{q_\mathrm{max} - 1}
    \]
    for all $j \in G$ and simply check that $(x, p)$ is an HZ equilibrium. 
\end{proof}

As a corollary we of course get that HZ earnings equilibria always exist.

\begin{theorem}
    An HZ earnings equilibrium always exists.
\end{theorem}

\section{Bivalued Utilities}\label{sec:bivalued}
An interesting special case in the all-goods setting is that in which $u_{i j} \in \{0, 1\}$ for all $i, j$. Such utilities are called \emph{dichotomous} and are of particular interest. Furthermore, they represent the only special case (with unbounded agents / goods) in which HZ equilibria are known to be
polynomial time computable \cite{VY21}.
Moreover, due to a result of Bogomolnaia and Moulin \cite{BM04}, the HZ mechanism is also incentive compatible in the dichotomous setting. 

We can easily extend these results to more generalized bivalued utilities, including chores. Such utilities have been extensively studied (see, e.g.,~\cite{BM04,VY21,BabaioffEF21,GargMQ22,EbadianPS22}) due to their practical relevance. In real-world scenarios, it may be easier for agents to express their preferences by simply categorizing their desire for an item as ``high'' or ``low'' using numerical values, rather than providing a detailed utility function. 

\begin{theorem}
    If the utilities are of the form $u_{i j} \in \{a_i, b_i\}$ with $(a_i)_{i \in A}$ and $(b_i)_{i
    \in B}$ rational, we can compute HZ equilibria in polynomial time. Moreover, the resulting
    mechanism is strategyproof.
\end{theorem}

\begin{proof}
    As per Section~\ref{sec:shifting}, we can shift the agents' utilities until they are of the form
    $u_{i j} \in \{0, c_i\}$ for some non-negative vector $c$.
    Moreover, we can rescale the utilities so that in fact $u_{i j} \in \{0, 1\}$.
    So we can run the algorithm by Vazirani and Yanakakis \cite{VY21} to obtain an HZ equilibrium
    and the resulting mechanism is strategy-proof by the result of Bogomolnaia and Moulin
    \cite{BM04}. 
\end{proof}

Due to Theorem~\ref{thm:equilibria_equiv}, we can also compute HZ earnings equilibria in polynomial time for bivalued instances.

\section{Constantly Many Agent Types}\label{sec:const_types}

Another interesting special case is the setting in which there are constantly many types of agents,
or alternatively in which there are constantly many agents with different demands.
Here, we are given a small set $A$ of agents and a set $G$ of goods or chores.
Each agent $i \in A$ has some demand $d_i$ such that $\sum_{i \in A} d_i = |G| = n$.
The goal is to fractionally assign $G$ to $A$ such that every agent $i$ gets exactly $d_i$ units of
goods and chores.

\begin{definition}
    In the setting with demands, an FPM $x$ is considered \emph{envy-free}, if for every agent $i$
    and $i'$ we have
    \[
        \frac{u_i \cdot x_i}{d_i} \geq \frac{u_i \cdot x_{i'}}{d_{i'}}.
    \]
    In other words, agents experience envy on the basis of utility per demand.
\end{definition}

Note that this setting includes the setting in which there are many agents but which only have $k$ different utility vectors.
For each unique utility vector $(v_j)_{j \in G}$, we create a representative agent $i$ with demand $d_i = l$ where $l$ is the number of agents with $u_i = v$.
Given an EF+PO allocation in this contracted instance, we can create an EF+PO allocation in the original instance by distributing equal shares of each agent $i$'s bundle to the $d_i$ many agents that it represents.
The definition of envy-freeness given above is chosen so that envy-freeness is maintained by this reduction.

The setting with constantly many agents can be approached in two different ways.
First, Alaei et al.\ \cite{AJT17} provides an algorithm that computes
HZ equilibria with a constant number of agents. By our observation that HZ equilibria can be
shifted, this implies that one can also find HZ equilibria in the presence of chores when there are
constantly many agents.
The downside of this approach is that it is very slow as it relies on algebraic cell decomposition;
for $k$ agents, their algorithm checks on the order of $k^{5 k^2}$ many cells and this check is in
itself not trivial either.
Additionally, HZ equilibria are not generally rational so the output of the algorithm is a
description of the solution in terms of roots of polynomials.
Garg et al.\ \cite{GargTV22} give an alternative algorithm to find $\epsilon$-approximate HZ equilibria by solving $O(n k^k \epsilon^{-2 k})$ linear programs.
This approach, while in principle more efficient than Alaei et al.\ \cite{AJT17}, is also intractible in practice.

On the other hand, one may consider EF+PO allocations.
We are not aware of any algorithms for computing EF+PO allocations with constantly many agent types
which are more efficient than finding an HZ equilibrium via algebraic cell decomposition.
If the entire instance has constant size, i.e.\ both the types of agents \emph{and} of goods are
constant, then a simple polyhedral approach can be used.
In that case, the set of envy-free allocations forms a polytope in constant dimension and as shown
by Tr\"obst and Vazirani \cite{TV24}, there is always a vertex of this EF polytope which is EF+PO.
By enumerating the vertices of the EF polytope, it is hence possible to efficiently find EF+PO allocations.

Lastly, let us consider an even more special case, namely the case in which there are only two types of agents, where each agent's utility function is one of the two given functions. This scenario has been extensively studied across various related contexts (see, e.g.,~\cite{AzizLRS23,GargMQ24}) In this setting, it is possible to find an EF+PO allocation in polynomial time via linear programming. 
Our approach relies on the following lemma.

\begin{lemma}\label{lem:2_agents}
    Assume there are only two types of agents with different demands, i.e.\ $|A| = 2$, and let $x$ be an envy-free FPM.
    Moreover, let $y$ be another FPM which is Pareto-better than $x$.
    Then $y$ is also envy-free.
\end{lemma}

\begin{proof}
    Note that for any good $j \in G$, we have $x_{1 j} + x_{2 j} = 1$ since $x$ is a FPM and
    likewise we have $y_{1 j} + y_{2 j} = 1$.
    In particular, this implies that $y_{1 j} - x_{1 j} = x_{2 j} - y_{2 j}$.
    Hence we have that
    \begin{align*}
        u_1 \cdot y_2 &= u_1 \cdot x_2 + u_1 \cdot (y_2 - x_2) \\
        &= u_1 \cdot x_2 + u_1 \cdot (y_1 - x_1) \\
        &\geq u_1 \cdot x_2
    \end{align*}
    using the fact that $u_1 \cdot y_1 \geq u_1 \cdot x_1$ since $y$ is Pareto-better than $x$.
    In other words, agent 1 thinks that agent 2's bundle is worse in $y$ than in $x$.

    Now it is easy so see that agent 1 cannot envy agent 2:
    \begin{align*}
        \frac{u_1 \cdot y_1}{d_1} &\geq \frac{u_1 \cdot x_1}{d_1} \\
        &\geq \frac{u_1 \cdot x_2}{d_2} \\
        &\geq \frac{u_1 \cdot y_2}{d_2}.
    \end{align*}
    By symmetry, agent 2 also cannot envy agent 1. 
\end{proof}

Consider now the following LP.
\begin{maxi*}
    {}
    {\sum_{i \in A} u_i \cdot x_i}
    {}
    {}
    \addConstraint{u_i \cdot x_i}{\geq u_i \cdot x_{i'}}{\quad \forall i \in A, i' \in A \setminus
    \{i\}}
    \addConstraint{\sum_{j \in G} x_{i j}}{= d_i}{\quad \forall i \in A}
    \addConstraint{\sum_{i \in A} x_{i j}}{= 1}{\quad \forall j \in G}
    \addConstraint{x_{i j}}{\geq 0}{\quad \forall i \in a, j \in G.}
\end{maxi*}
Its optimal solution $x$ is Pareto-optimal \emph{among the envy-free solutions}.
Note that in general, $x$ need not be Pareto-optimal \emph{among all solutions} which is what we
generally refer to as PO.
However, by Lemma~\ref{lem:2_agents}, $x$ is indeed PO in the setting with two types of agents and so this
gives a simple way of finding EF+PO solutions in polynomial time.

Finally, let us briefly mention the setting in which there are constantly many items instead of agents.
This is a much simpler scenario as one can enumerate the price vectors on the goods via a simple grid search in order to find an approximate HZ equilibrium in polynomial time.
Moreover, the special case of two types of goods simply reduces to the bivalued setting discussed in Section~\ref{sec:bivalued}.

\section{Nash Bargaining}\label{sec:nash_bargaining}

Aside from instances with constantly many agent or good types \cite{AJT17}, the bivalued case is in
fact the only case in which a polynomial time algorithm is known for the goods version of HZ.
Indeed, finding an approximate HZ equilibrium is PPAD-hard even with only four utility values
\cite{CCPY22}.

Tr\"obst and Vazirani \cite{TV24} recently showed that in fact finding any allocation which is
envy-free and Pareto-optimal is already PPAD hard.
However, they also show that the alternative, Nash-bargaining-based mechanism proposed by Hosseini
and Vazirani \cite{HV22} is 2-approximately envy-free and 2-approximately incentive compatible.
The idea behind this mechanism is to maximize Nash welfare, i.e.\ the product of agents' utilities.
An Nash bargaining allocation is therefore given by the solution to:

\begin{maxi*}
    {}
    {\prod_{i \in A} u_i \cdot x_i}
    {}
    {}
    \addConstraint{\sum_{j \in G} x_{i j}}{= 1}{\quad \forall i \in A}
    \addConstraint{\sum_{i \in A} x_{i j}}{= 1}{\quad \forall j \in G}
    \addConstraint{x_{i j}}{\geq 0}{\quad \forall i \in A, j \in G.}
\end{maxi*}

For non-negative utilities, the objective function is strictly non-decreasing in the utility of all
the agents and therefore the Nash bargaining solution is Pareto-optimal.
Moreover, the objective function is log-concave and therefore we can find Nash bargaining solutions
very efficiently in both theory and practice \cite{PTV21}.

The attractive game theoretic properties together with the polynomial time computability make Nash
bargaining a promising alternative to HZ for the all goods setting.
We therefore pose the question: is there an analogous mechanism for the chores or even mixed settings?
Obviously, the mechanism as stated does not produce meaningful results when some or all of the
utilities can be negative.

Note that unlike for exact EF+PO or HZ allocations, shifting utilities does not make sense when the
goal is to find \emph{approximately} EF+PO allocations.
To see why, consider an instance where all $u_{ij} \in [0, 1]$.
If we simply shift the instance by say $c = 100$, then \emph{any} solution is $1.01$-EF.
In particular, we do not get meaningful guarantees if we shift utilities from negative to positive
in order to apply existing approximate mechanisms.

There are however two natural ways in which the Nash-bargaining-based mechanism can be generalized
to the chores setting.
The first is to minimize the product of agents' \emph{disutilities}, i.e. the negation of their
utilities.
The allocation would therefore be given by
\begin{mini*}
    {}
    {\prod_{i \in A} (-u_i \cdot x_i)}
    {}
    {}
    \addConstraint{\sum_{j \in G} x_{i j}}{= 1}{\quad \forall i \in A}
    \addConstraint{\sum_{i \in A} x_{i j}}{= 1}{\quad \forall j \in G}
    \addConstraint{x_{i j}}{\geq 0}{\quad \forall i \in A, j \in G.}
\end{mini*}

The resulting allocation $x$ is clearly Pareto-optimal.
Unfortunately, there can be large amounts of envy between agents.
For an example, see Figure~\ref{fig:nash1}.
Here we have two agents $i$ and $i'$ who both agree that chore $j'$ is worse than chore $j$.

Let $t$ be the amount that $i$  is matched to $j'$, i.e. $t = x_{i j'}$.
Then, the objective under Nash bargaining would be to minimize $(C t + 1 - t) (1 - t)$ which is
achieved when $t = 1$.
The issue is that this allocation creates a large amount of envy.
From the perspective of agent $i$, the bundle assigned to agent $i$ has disutility $C$ whereas the
bundle assigned to agent $i'$ has disutility 1.
Therefore, agent $i$ has factor $C$ envy towards agent $i'$.

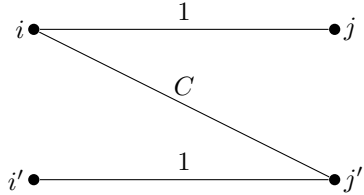
\begin{figure}[htb]
    \centering
    \begin{tikzpicture}
        \node[fill, circle, inner sep=1.5pt] (v2) at (0, 0) {};
        \node[fill, circle, inner sep=1.5pt] (v1) at (0, 2) {};
        \node[fill, circle, inner sep=1.5pt] (w2) at (4, 0) {};
        \node[fill, circle, inner sep=1.5pt] (w1) at (4, 2) {};

        \draw[-] (v1) -- (w1) node[midway,above] {$1$};
        \draw[-] (v1) -- (w2) node[midway,above] {$C$};
        \draw[-] (v2) -- (w2) node[midway,above] {$1$};

        \node[left] at (v1) {$i$};
        \node[left] at (v2) {$i'$};
        \node[right] at (w1) {$j$};
        \node[right] at (w2) {$j'$};
    \end{tikzpicture}
    \caption{This example shows that minimizing the product of disutilities does not lead to fair
    outcomes. Here, agent $i$ has utilities $(1, C)$ and agent $i'$ has utilities $(0, 1)$. Both
    agents agree that $j'$ is the worse chore. However, the Nash bargaining solution will assign
    $j'$ fully to $i$. \label{fig:nash1}}
\end{figure}

Minimizing the product of disutilities, while Pareto-optimal, does evidently not lead to fair outcomes.
For an alternative generalization, we look towards the work of Bogomolnaia et al.\ \cite{BMSY17} who
studied the traditional, divisible fair division problem \emph{without} the matching constraint.
In this setting, it is well-known that Nash bargaining solutions actually coincide with Fisher
market equilibria which (for equal budgets) are always envy-free.

What Bogomolnaia et al.\ showed is that in the chores setting, the correct generalization of Nash
bargaining which lines up with competitive equilibria and hence produces envy-free allocations, is
to \emph{maximize} the product of the agents' disutilities over the set of Pareto-optimal
allocations.
We call this \emph{Pareto-constrained Nash bargaining}.
In our model which includes the matching constraint, this would mean solving the program:
\begin{maxi*}
    {}
    {\prod_{i \in A} (-u_i \cdot x_i)}
    {}
    {}
    \addConstraint{x \text{ is Pareto-optimal}}
    \addConstraint{\sum_{j \in G} x_{i j}}{= 1}{\quad \forall i \in A}
    \addConstraint{\sum_{i \in A} x_{i j}}{= 1}{\quad \forall j \in G}
    \addConstraint{x_{i j}}{\geq 0}{\quad \forall i \in A, j \in G.}
\end{maxi*}

Note that it is not clear if this can even be solved in polynomial time, though recent work provides
some indication that this may be possible. \cite{CGMM22,CKMN24}
Unfortunately, it turns out that for matchings, this too does not yield fair outcomes.
This is demonstrated in Figure~\ref{fig:nash2}.

Once again, let $t = x_{i j'}$ be the amount of $j'$ that $i$ receives.
Note that any fractional perfect matching is Pareto-optimal in this instance.
The product of the agents' disutilities is $(2t + 1 - t) (1 - t) = 1 - t^2$ which is clearly
maximized when $t = 0$.
But that means that, from the perspective of agent $i'$, $i'$ has disutility 1 whereas $i$ has
disutility 0.
Again, this causes unbounded envy from $i'$ towards $i$.

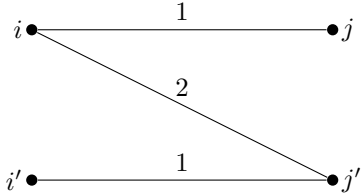
\begin{figure}[htb]
    \centering
    \begin{tikzpicture}
        \node[fill, circle, inner sep=1.5pt] (v2) at (0, 0) {};
        \node[fill, circle, inner sep=1.5pt] (v1) at (0, 2) {};
        \node[fill, circle, inner sep=1.5pt] (w2) at (4, 0) {};
        \node[fill, circle, inner sep=1.5pt] (w1) at (4, 2) {};

        \draw[-] (v1) -- (w1) node[midway,above] {$1$};
        \draw[-] (v1) -- (w2) node[midway,above] {$2$};
        \draw[-] (v2) -- (w2) node[midway,above] {$1$};

        \node[left] at (v1) {$i$};
        \node[left] at (v2) {$i'$};
        \node[right] at (w1) {$j$};
        \node[right] at (w2) {$j'$};
    \end{tikzpicture}
    \caption{This example shows that maximizing the product of disutilities over the set of
    Pareto-optimal allocations does not lead to fair outcomes. Agent $i$ has utilities $(1, 2)$
    whereas agent $i'$ has utilities $(0, 1)$. The Pareto-constrained Nash bargaining solution gives
    $j$ to $i$ and $j'$ to $i'$.\label{fig:nash2}}
\end{figure}

As we have demonstrated, neither Nash bargaining nor Pareto-constrained Nash bargaining allow us to achieve constant factor approximations of envy-freeness in the context of chores.

\section{Discussion}

In this paper, we initiated the study of matching markets involving chores. We showed that, due to the matching constraints, many results from the goods setting can be extended to mixed or chores settings via utility shifting. However, this approach does not apply to the problem of finding approximately envy-free and Pareto-optimal allocations. This leads us to what we believe is an exciting open question for cardinal-utility matching markets with chores:

\begin{question}
    Is there a polynomial time algorithm that finds approximately fair and efficient allocations in a cardinal-utility matching market with chores?
\end{question}

We remark that for the more general mixed setting, it is not entirely obvious what an appropriate notion of approximate fairness is.

%\newpage

\bibliographystyle{plain}
\bibliography{references}

\end{document}